\documentclass[letterpaper,11pt]{article}
\usepackage{fullpage,epsf,epsfig,latexsym,amsmath,amsthm,amssymb,url}

\newtheorem{theorem}{Theorem}[section]
\newtheorem{proposition}[theorem]{Proposition}
\newtheorem{corollary}[theorem]{Corollary}
\newtheorem{lemma}[theorem]{Lemma}

\newcommand{\remove}[1]{}
\newcommand{\F}{\vspace*{\smallskipamount}}
\newcommand{\FF}{\vspace*{\medskipamount}}
\newcommand{\FFF}{\vspace*{\bigskipamount}}

\newcommand{\T}{\hspace*{2em}}
\newcommand{\TT}{\hspace*{4em}}
\newcommand{\TTT}{\hspace*{6em}}

\newcommand{\diff}[2]{\frac{\partial #1}{\partial #2}}
\newcommand{\bra}[1]{\left(#1\right)}

\renewcommand{\Pr}[1]{\Prsymbol\left[#1\right]}
\newcommand{\Exp}[1]{\Expsymbol\left[#1\right]}

\newcommand{\card}[1]{\left|#1\right|}
\newcommand{\set}[1]{\left\{#1\right\}}

\newcommand{\Prsymbol}{\mathop\mathsf{Pr}}
\newcommand{\Expsymbol}{\mathop\mathsf{Exp}}

\renewcommand{\a}{\mbox{$\cal A$}}

\newcommand{\g}{\mbox{$\cal G$}}

\newcommand{\fa}{\mbox{{\scriptsize $\cal A$}}}

\newcommand{\fg}{\mbox{{\scriptsize $\cal G$}}}

\begin{document}

\title{Scheduling Dags under Uncertainty\thanks{Preliminary version of this work appeared in the Proceedings of the 17th ACM Symposium on Parallelism in Algorithms and Architectures (SPAA'05). Research performed in part during a visit to the Division of Mathematics and Computer Science, Argonne National Laboratory, Argonne, IL 60439 supported by NSF grant ITR-800864, and a stay with the Department of Computer Science, University of Alabama, Tuscaloosa, AL 35487.}
\author{
    Grzegorz Malewicz\\
    \small Department of Engineering\\
    \small Google, Inc.\\
    \small Mountain View, CA 94043, USA\\
    \small malewicz@google.com\\
}
}

\maketitle

\begin{abstract}
This paper introduces a parallel scheduling problem where a directed acyclic graph modeling $t$ tasks and their dependencies needs to be executed on $n$ unreliable workers. Worker $i$ executes task $j$ correctly with probability $p_{i,j}$.
The goal is to find a regimen $\Sigma$, that dictates how workers get assigned to tasks (possibly in parallel and redundantly) throughout execution, so as to minimize the expected completion time. 
This fundamental parallel scheduling problem arises in grid computing and project management fields, and has several applications.

We show a polynomial time algorithm for the problem restricted to the case when dag width is at most a constant and the number of workers is also at most a constant. These two restrictions may appear to be too severe. However, they are fundamentally required. Specifically, we demonstrate that the problem is NP-hard with constant number of workers when dag width can grow, and is also NP-hard with constant dag width when the number of workers can grow. When both dag width and the number of workers are unconstrained, then the problem is inapproximable within factor less than $5/4$, unless P=NP.
\end{abstract}


\section{Introduction}

Grid computing infrastructures have been developed over the past several years to enable execution of computations on shared distributed resources~\cite{FosterK04}. The machines, disks and network often operate at a slower pace or stop operating, due to hardware and software failures and sharing. Nevertheless, there is significant demand to perform scientific computations with complex task dependencies on grids (cf. \cite{Ann02,Ber+03,Inspiral,H+05}). Among the most important remaining challenges is to determine how to quickly execute large-scale, sophisticated computations using unreliable resources. When a task fails to get computed correctly, then the progress of execution may be delayed because dependent tasks cannot be executed pending successful execution of the task.
It is conceivable that task dependencies and resource reliabilities play a significant role in the ability to execute computations quickly. Therefore, one would like to determine relationships among these factors, and develop algorithms for quick execution of complex computations using unreliable resources.

A similar problem arises when managing projects~\cite{HL01} such as production planning or software development. Here a collection of activities and precedence constraints are given. Workers can be assigned to perform the activities. In practice, a worker assigned to an activity may fail to perform it. For example, if an activity consists of writing a piece of code and testing it, it could happen that the test fails. The manager of the project may be able to estimate the success probability of a worker assigned to an activity based on prior experience with the worker. The manager may be able to redundantly assign workers to an activity. For example, two workers may independently write a piece of code and test it; if at least one test succeeds, the activity is completed. Thus the manager faces a problem of how to assign workers to activities, possibly in parallel and redundantly, over the course of the project, so as to minimize the total time of conducting the project.



These two application areas motivate the study of the following fundamental parallel computing scheduling problem. A directed acyclic graph (dag) is given representing $t$ {\em tasks} and their dependencies. There are $n$ {\em workers}. At any given unit of time workers are assigned in some way to the tasks that are ``eligible'' based on precedence constraints and tasks executed thus far (any worker is assigned to at most one task at a time; more than one task may get assigned; more than one worker can be assigned to a task; workers can idle). The workers then {\em attempt} to execute the assigned tasks. The attempt of worker $i$ to execute task $j$ {\em succeeds} with probability $0\leq p_{i,j} \leq 1$. In the next unit of time workers are again assigned to tasks. The execution proceeds in this manner until the $t$ tasks have been executed. The goal is to determine a regimen $\Sigma$, that dictates how workers get assigned to tasks throughout execution, that minimizes the expected completion time.


\F
\subsection{Contributions}
In this paper we propose and investigate a parallel scheduling problem of executing dags using unreliable workers. Our specific contributions are as follows:
\begin{itemize}
\item[(i)] We introduce a new combinatorial optimization problem: given a dag $\g$ with $t$ tasks and $n$ workers such that $p_{i,j}$ is the probability that worker $i$ executes task $j$ correctly, find a regimen $\Sigma$ that minimizes the expected completion time.
\end{itemize}

We show that the ability to solve this problem in polynomial time depends in a crucial way on two natural parameters of the problem: the number of workers and the {\em width} of the dag $\g$. The latter parameter denotes the maximum cardinality of an {\em antichain} (all technical terms are defined in Section~\ref{s.def}).

\begin{itemize}
\item[(ii)] We develop a polynomial time algorithm that finds an optimal regimen for a restricted version of the problem. Specifically, we assume that the width of the dag $\g$ is at most a constant, {\em and also} that the number of workers is at most a constant. Note that our algorithm allows the dag to have complex structure and dag size to grow, however the dag must be ``narrow''. The algorithm uses a dynamic programming approach. Given dag $\g$, we construct a directed graph $\a$ called {\em admissible evolution of execution} that contains all possible sets of tasks that a regimen can have executed at a point of time, and how one set could result form another within one unit of time. This graph $\a$ turns out to be a dag. We then formulate a recurrence to define a regimen for a node of $\a$ based on definitions for the children of the node. 

\item[(iii)] We sharply contrast our algorithmic results with complexity lower bounds, by demonstrating that our restrictions are fundamentally necessary. Specifically, we show that the optimization problem is NP-hard when dag width is constant while the number of workers can grow, and is also NP-hard when the number of workers is constant while dag width can grow. Moreover, if both dag width and the number of workers are unconstrained, then we show that the problem cannot be approximated with a factor less than $5/4$, unless P=NP.
\end{itemize}

Our reductions demonstrate that the complexity of the problem is quite steep. 
First, the problem is trivial when the dag has just one task. Simply then an optimal regimen assigns all workers to that task. However, the problem becomes NP-hard when dag has just two independent tasks! (Arguably this is the ``second simplest'' dag.) 
Second, the problem is easy when there is just one worker. Simply then an optimal regimen follows any topological sort and assigns the (only) worker to the task at hand. However, the problem becomes NP-hard when there are just two workers! 

\F\F
\subsection{Related work}
There are several studies that deal with scheduling under constrained resources~\cite{OU95,NR01}. Many such problems are NP-hard~\cite{CK04}, and the studies typically focus on heuristic approaches. Our work differs from the studies because we also consider resource failures.
More closely related are results~\cite{MT97} on scheduling where each task can be executed in one of a few ways. Each way has certain resource requirements, duration, and failure probability. After execution failure, the task is reexecuted fault-free. The goal is to minimize the expected completion time. Several heuristics are proposed. A similar model is studied~\cite{TMY95} with the goal of maximizing the probability of successful completion using heuristics.

In project scheduling under uncertainty~\cite{HL05,FAP98a,FAP98b} each task has a duration that is a random variable and a requirement on the amount of different resources needed to perform the task. There is a fixed amount of each resource available. The goal is to determine how to assign resources to tasks over the execution of the project to minimize the expected completion time. Our problem is similar because workers can be treated as a resource. However, our problem is different because we allow multiple workers to be assigned to a task, which modifies the distribution of task duration.
A model is studied~\cite{TN03} where the execution time of each task in a dag is a random variable with a certain distribution that depends on the amount of resources assigned to the task. Execution may fail if it takes longer than a threshold. The goal is to maximize the probability of completion. Several heuristics have been proposed.

One of the goals of stochastic scheduling is to minimize the expected completion time when task durations are random variables and tasks need to be scheduled on parallel machines. Such problem was studied for independent tasks~\cite{KRT00,GI99} and for dependent tasks~\cite{SU01}, however, under a different setting, because we assume that two or more machines may be assigned to the same task, which may modify the probability distribution of task duration. 
Other related problems are the Network Reliability Problem and the Network Survivability Problem~\cite{GJ79} that model failures probabilistically, but have different optimization goals.

\remove{
Several interesting results on robust scheduling are related to our work (cf.~\cite{AMSK04,STW98,KY96}). Here the goal is to derive schedules whose performance according a given metric is controlled when the input differs from a specific input in a controlled manner.
}

Computing tasks over the Internet poses distinct challenges.
It is known~\cite{GaoM04} which dependent tasks should be executed by a reliable server and which by workers whose unreliability is modeled probabilistically, so as to maximize the expected number of correct tasks, under a constraint on completion time.
A similar probabilistic model is studied~\cite{Sar02} for independent tasks.
Also related is the problem of scheduling tasks so as to render tasks eligible for allocation to workers (hence for execution) at the maximum possible rate~\cite{Rosenberg04,RosenbergY04,MRY05,MalewiczR04,MFRW07,CMR07}.

\remove{
Pebble Games have been used to model scheduling problems on dags~\cite{PatersonH70,Cook74,HopcroftPV77,HongK81,RosenbergS83}. In this model we are given a supply of pebbles and certain rules that restrict how pebbles can be placed on the dag's nodes. The goal is to determine how to place the pebbles on the nodes (i.e., how to ``play the game'') to meet a certain performance goal. 
}

There are several systems in grid computing and project management fields that are related to our work. Condor~\cite{TTL05}, for example, executes computations with complex task dependencies. The clients to which tasks are sent are commonly unreliable. Condor assigns any task to one computer at a time (no redundancy), and uses a
``FIFO'' topological sort to sequence task submissions. This may sometimes lead to an ineffective use of computing resources.
One example of a project management system is the Microsoft Project 2003~\cite{MsftProj03}. The system can estimate project duration using a PERT algorithm based on a probabilistic model of duration of a task, but does not take into account resource constraints.

In a companion to this paper the author explores~\cite{Mal05b} implementation details and experiments with the algorithm. Specifically, a range of theoretical and practical approaches are proposed to craft an efficient implementation of the algorithm. The benefits of the approaches and scalability of the implementation are evaluated experimentally.

Approximation algorithms for several cases of the problem studied in this paper were recently given~\cite{LR07,L07}. The case of independent tasks was shown to admit an $O(\log t)$-approximation. Richer dependency structures also admit approximations. For disjoint chains, an
$O(\log n \log t \log(n+t)/\log\log(n+t))$-approximation was given. A collection of directed out- or in-trees admits an $O(\log n \log^2t)$-approximation, while a directed forest admits an $O(\log n \log^2 t \log(n+t)/\log\log(n+t))$-approximation.

\subsection{Paper organization}
The remainder of paper is structured as follows. In Section~\ref{s.def}, we define a model of executing dags where workers can fail with certain probabilities, and formulate an optimization problem of minimizing the expected completion time. In Section~\ref{s.alg}, we give a polynomial time algorithm for the problem with dags that have constant width and where the number of workers is constant. Finally, in Section~\ref{s.complexity}, we explain why restricting dag width and the number of workers is fundamentally required.

\section{Model and background}

\label{s.def}

A {\it directed graph} $\g$ is given by a set of {\it nodes} $N_{\fg}$
and a set of {\it arcs} (or, {\it directed edges}) $A_{\fg}$, each
having the form $(u \rightarrow v)$, where $u, v \in N_{\fg}$.  A {\it
path} in $\g$ is a sequence of arcs that share adjacent endpoints, as
in the following path from node $u_1$ to node $u_k$: $(u_1 \rightarrow
u_2), \ (u_2 \rightarrow u_3), \ \ldots, \ (u_{k-2} \rightarrow
u_{k-1}), \ (u_{k-1} \rightarrow u_k)$.  A {\it dag} ({\it directed
acyclic graph}) $\g$ is a directed graph that has no cycles; i.e., in
a dag, no path of the preceding form has $u_1 = u_k$.
Given an arc $(u \rightarrow v) \in A_{\fg}$, $u$ is a {\em parent} of
$v$, and $v$ is a {\em child} of $u$ in $\g$.  Each parentless node of
$\g$ is a {\em source (node)}, and each childless node is a {\em sink
(node)}; all other nodes are {\em internal}.

Given a dag, an {\em antichain} is a set of its nodes such that no two are ``comparable'' i.e., for any two distinct nodes $u$ and $v$ from the set, there is no path from $u$ to $v$ nor from $v$ to $u$. The largest cardinality of an antichain is called {\em width} of the dag.
A {\em chain} is a path. A set of chains is said to {\em cover} the dag if every node of the dag is a node in at least one of the chains (chains may ``overlap''). A Dilworth's Theorem~\cite{Dil50} states that dag width is equal to the minimum number of chains that cover the dag.

A computation is modeled by a dag. Then the nodes are called {\em tasks} and for concreteness we assume that $N_{\fg}=\set{1,\ldots,t}$. We denote a set $\set{1,\ldots,t}$ by $[t]$. Arcs specify dependencies among tasks: given an arc $(u \rightarrow v)$, $v$ cannot be executed until $u$ has been. A set of tasks {\em satisfies precedence constraints} if, for every task in the set, every parents of the task is also in the set. Given such set $X$, we denote by $E(X)$ the set of tasks not in $X$ every of whose parent is in $X$; tasks in this set are called {\em eligible} when tasks $X$ have been executed. (So any source not in $X$ is eligible.)

The execution of tasks is modeled by the following game.
There are $n$ {\em workers} identified with elements of $[n]$. Let $X$ be a set of tasks that satisfies precedence constraints. The game starts with $Y=X$, and proceeds in {\em rounds}. During a round, workers are assigned to tasks in $E(Y)$ according to a regimen $\Sigma$. The regimen specifies an assignment $\Sigma(Y)$ that maps each worker to an element of the set $E(Y)\cup\set{\perp}$ i.e., either to a task that is eligible in this round, or to a distinguished element $\perp$. Note that the assignment is determined by the set of tasks $Y$. The assignment enables directing multiple workers to the same task, or to different tasks; workers can also idle. 
Then each worker that was assigned to a task {\em attempts} to execute the task.
The attempt of worker $i$ assigned to task $j$ {\em succeeds} with probability $0\leq p_{i,j} \leq 1$ independently of any other attempts. We assume that there is at least one worker that has non-zero probability of success, for any given task. A task is {\em executed} in this round if, and only if, at least one worker assigned to the task has succeeded. Every executed task is added to $Y$, and the game proceeds to the next round. It could be the case that every attempt has failed; then the set $Y$ remains unchanged, and in the next round worker assignment remains unchanged, too.
Formally, a regimen $\Sigma$ is a function $\Sigma:2^{N_{\fg}}\rightarrow ([n] \rightarrow (N_{\fg}\cup\set{\perp}))$, such that for any subset $Z$ of tasks that satisfies precedence constraints, the value $\Sigma(Z)$ is a function from $[n]$ to the set $E(Z)\cup\set{\perp}$.
The game proceeds until a round when every sink of $\g$ is in $Y$. We say that the game {\em ends} in such round.

The quality of the game is determined by how quickly the game ends. Specifically, the number of rounds of the game, beyond the first round, until the round when the game ends is called {\em time to completion} of regimen $\Sigma$ starting with tasks $X$ already executed. This time is a random variable. When $X$ is empty, we call the time simply {\em completion time}.
Our goal is to find a regimen $\Sigma_{OPT}$ that minimizes the expected completion time. We call this goal the Recomputation and Overbooking allowed Probabilistic dAg Scheduling Problem (ROPAS).

\FF
\noindent
{\bf ROPAS}\\
{\it Instance:} A dag $\g$ describing dependencies among $t$ tasks, $n$ workers such that worker $i$ succeeds in executing task $j$ with probability $0\leq p_{i,j}\leq 1$, and that for any task $j$ there is worker $i$ with $p_{i,j}>0$.\\
{\it Objective:} Find a regimen $\Sigma_{OPT}$ that minimizes the expected completion time.
\FF

We observe that the optimization problem yields a finite expectation.
Let $X$ be a subset of tasks that satisfies precedence constraints. Denote by $B_X$ the minimum expected time to completion across regimens that start with tasks $X$ already executed. 
The subsequent lemma states that $B_X$ is finite. 
The proof, while simple, is presented here in detail to help the reader gain familiarity with our notation. 

\begin{lemma}
\label{l.finite}
For any set $X$ of tasks that satisfies precedence constraints, $B_X$ is finite.
\end{lemma}

\remove{**********************
\begin{proof}[sketch]
We take a topological sort, and assign all workers to successive unexecuted tasks. This yields a finite expectation, so $B_X$ must be finite, too.
\end{proof}
**********************}

\begin{proof}
We define a regimen that has finite expectation; thus minimum expectation must be finite, too. We take a topological sort $t_1,\ldots,t_m$ of the subdag of $\g$ induced by tasks not in $X$. Note that it is possible to execute tasks in the order of the sort because once tasks in $X$ and $t_1,\ldots,t_{j-1}$ have been executed, task $t_j$ is eligible. Our regimen will follow this order, each time assigning every worker to the task at hand. The probability that every worker fails to execute task $t_j$ is $\prod_{i=1}^n (1-p_{i,t_j})$, which by assumption is strictly smaller than $1$. Thus the expected time to execute the task is $1/(1-\prod_{i=1}^n (1-p_{i,t_j}))$, because execution time follows geometric distribution. We can use linearity of expectation to conclude that the expected time to completion for the regimen is just the sum of expectations for each individual task in the sort. Thus $B_X$ is at most $\sum_{j=1}^t1/(1-\prod_{i=1}^n (1-p_{i,j}))$, as desired. 
\end{proof}

\section{Scheduling algorithm for a restricted problem}
\label{s.alg}

This section presents a polynomial time algorithm that finds an optimal regimen for a restricted version of the ROPAS Problem. Specifically, we make two assumptions: (1) the dag $\g$ has at most a constant width, and (2) the number of workers is at most a constant. These two restrictions may appear to be quite severe, and perhaps too restrictive! However, they are fundamentally necessary, as we shall see in Section~\ref{s.complexity}.

We construct a directed graph that models how computation can evolve. The graph $\a=(N_{\fa},A_{\fa})$ called {\em admissible evolution of execution} for $\g$ (an example is in Figure~\ref{f.aec}) is constructed inductively. Each node of $\a$ will be a subset of nodes of $\g$. We begin with a set $ N_{\fa}=\set{\emptyset}$. For any node $X \in N_{\fa}$ that does not contain every sink of $\g$, we calculate the set of eligible tasks $E(X)$ in $\g$. We then take the non-empty subsets $D \subseteq E(X)$, add to $N_{\fa}$ a node $X\cup D$, if it is not there already, and add to $A_{\fa}$ an arc $(X,X\cup D)$, if it is not there already. Since $\g$ is finite, the inductive process clearly defines a unique directed graph $\a$. The structure of this graph is explained in the subsequent lemma.

\begin{lemma}
\label{l.dag}
Let $\a$ be the admissible evolution of execution for a dag $\g$. Then $\a$ is dag. Its nodes are exactly the sets of tasks that satisfy precedence constraints. It has a single source $\emptyset$ and a single sink $N_{\fg}$.
\end{lemma}

\begin{proof}
We verify the assertions in turn.
The graph $\a$ cannot have any cycle, because any arc points from a node $X$ to a node $X\cup D$, but $X\cup D$ has larger cardinality than the set $X$.

Nodes of $\a$ are exactly the sets of tasks that satisfy precedence constraints. Indeed, if $X$ satisfies precedence constraints, then clearly so does its union with any subset of $E(X)$. Thus any node of $\a$ satisfies precedence constraints. Now pick any subset $Y$ of tasks that satisfies precedence constraints. Let $t_1,\ldots,t_{\card{Y}}$ be a topological sort of the subdag of $\g$ induced by $Y$. Clearly, $t_j$ belongs to the set of tasks eligible when tasks $t_1,\ldots,t_{j-1}$ have been executed, for any $j$. So if $\set{ t_1,\ldots,t_{j-1}}$ is a node of $\a$, so is $\set{t_1,\ldots,t_{j}}$. Since $\emptyset$ is a node of $\a$, $Y$ must also be. A corollary to this is that $N_{\fg}$ is a node of $\a$.

We add a node $Y$ to $N_{\fa}$ only if there is an arc leading to $Y$ from some other node. So $Y$ cannot be a source. However, $\emptyset$ is a source because no arc leads to a set with the same number or fewer elements.

Pick any node $X$ of $\a$ and suppose that it does not contain the sinks of $\g$. By looking at a topological sort of $\g$ we notice that there is a task of $\g$ not in $X$ such that every parent in $\g$ of the task is in $X$. Thus $E(X)$ is not empty, and so $X$ has a child in $\a$. Pick any node $X$ of $\a$ that contains every sink of $\g$. Since $X$ satisfies precedence constraints, it contains every task, and so $X=N_{\fg}$. But $E(N_{\fg})$ is empty, so $X$ has no children in $\a$.
\end{proof}

Given a regimen $\Sigma$, we can convert $\a$ into a Markov chain that models how execution can evolve for this particular regimen. Specifically, for any node $X$ of $\a$, $\Sigma$ defines the assignment of workers to tasks from $E(X)$, thus yielding transition probabilities form $X$ back to $X$ and to the children of $X$ in $\a$. The expected time to completion is then just the expected hitting time of the sink of $\a$. We can use Markov chain theory to relate expected times to completion for the nodes of $\a$ as follows.

\begin{theorem}[\cite{Nor97}]
\label{t.recursive}
Consider a regimen $\Sigma$ and a set $X$ of tasks that satisfies precedence constraints and does not contain all sinks of $\g$.
Let $D_0,D_1,\ldots,D_k$ be the distinct subsets of $E(X)$, and $D_0=\emptyset$.
Let $a_i$ be the probability that $D_i$ is exactly the set of tasks executed by workers in the assignment $\Sigma(X)$. Let $X_i=X\cup D_i$. Then
$$
T_{X_0} = \frac{1}{a_1+\ldots+a_k}\bra{1+\sum_{i=1}^k a_i \cdot T_{X_i}}
~,
$$
where $T_{X_i}$ is the expected time to completion for regimen $\Sigma$ starting with tasks $X_i$ already executed, and by convention $1/0 = \infty$ and $0 \cdot \infty = 0$.
\end{theorem}

Our goal, however, is not to compute the expected completion time for a given regimen, but rather find a regimen that minimizes the expectation. For this purpose, we give a dynamic programming algorithm called OPT that defines a regimen called $\Sigma_{OPT}$. Since $\a$ has no cycles, we can hope to apply the recurrence of Theorem~\ref{t.recursive} starting from the sink and backtracking towards the source. Specifically, we initialize the regimen arbitrarily. Then we take a topological sort $Y_1,\ldots,Y_m$ of $\a$ and process it in the reverse order of the sort. See Figure~\ref{f.pseudocode} for a pseudocode. When we process a node $X$ of $\a$, we define two values: a number $S_X$ and an assignment $\Sigma_{OPT}(X)$.
We begin by setting $S_{Y_m}$ to $0$ and $\Sigma_{OPT}(Y_m)$ so that each worker is assigned to $\perp$. 
Now let $1\leq h <m$, and let us discuss how $S_{Y_h}$ and $\Sigma_{OPT}(Y_h)$ are defined. Let $D_0,\ldots,D_k$ be the distinct subsets of $E(Y_h)$, such that $D_0=\emptyset$. We consider the distinct $\card{E(Y_h)}^n$ assignments of the $n$ workers to the tasks of $E(Y_h)$, but not to $\perp$. For any assignment, we calculate the probability $a_i$ that $D_i$ is exactly the set of tasks executed by workers in the assignment. If $a_1+\ldots+a_k>0$, then we compute the weighted sum $1/(a_1+\ldots+a_k)\cdot(1+\sum_{i=1}^ka_iS_{Y_h \cup D_i} )$. We pick an assignment that minimizes the sum. We set $S_{Y_h}$ to the minimum, and $\Sigma_{OPT} (Y_h)$ to the assignment that achieves the minimum. Then we move back in the topological sort to process another node, by decreasing $h$. After the entire sort has been processed, the regimen $\Sigma_{OPT}$ has been determined.

\begin{figure}[!t]
\footnotesize 

Data structure: $T$ is a dictionary that maps nodes of $\a$ to distinct floating point variables\\

\begin{tabular}{ll}
  \begin{tabular}{l}
$OPT(t, \g ,n ,(p_{i,j}))$ \\
{\tt01}~ let $Y_1,\ldots,Y_m$ be a topological sort of $\a$\\
{\tt02}~ $S_{Y_m} = 0$\\
{\tt03}~ for $h=m-1$ downto $1$ do \\
{\tt04}~\T $min = \infty$\\
{\tt05}~\T for all assignments $A$ of workers to $E(Y_h)$\\
{\tt06}~\TT let $I\subseteq E(Y_h)$ be the assigned tasks\\
{\tt07}~\TT $sum =0$\\
  \end{tabular}
&
  \begin{tabular}{l}
{\tt08}~\TT for all nonempty subsets $D\subseteq I$\\
{\tt09}~\TTT let $a=\Pr{\text{$A$ executes exactly $D$}}$ \\
{\tt10}~\TTT $sum = sum + a\cdot S_{Y_h\cup D}$\\
{\tt11}~\TT let $a=\Pr{\text{ every assigned worker fails }}$ \\
{\tt12}~\TT if $(1+sum)/(1-a) < min$\\
{\tt13}~\TTT $min =  (1+sum)/(1-a)$\\
{\tt14}~\TTT $\Sigma_{OPT}(Y_h) = A$\\
{\tt15}~\T $S_{Y_h} = min$\\
  \end{tabular}
\end{tabular}

\caption{An algorithm for constructing an optimal regimen $\Sigma_{OPT}$, for a dag $\g$ describing dependencies among $t$ tasks, and $n$ workers such that worker $i$ executes task $j$ successfully with probability $p_{i,j}$.}
\label{f.pseudocode}
\end{figure}

Our goal is twofold. First, we show that the regimen $\Sigma_{OPT}$ indeed minimizes the expected completion time. We do this by considering an extended algorithm OPT$^\star$ where workers can idle. This freedom matches that defined in ROPAS and we can use the recurrence to demonstrate that OPT$^\star$ solves ROPAS. We then argue that idling is not needed, because an optimum can be found among regimens where every worker always gets assigned to a task. We, hence, prove that $OPT$ solves ROPSA, too. We conclude by showing that OPT runs in polynomial time when dag width and the number of workers are at most a constant.

It is convenient to consider an extended version OPT$^\star$ of the algorithm where workers are allowed to idle. Specifically, when processing $Y_h$, we consider the $(\card{E(Y_h)}+1)^n$ distinct assignments of the $n$ workers to the tasks of $E(Y_h)$ or to $\perp$. The regimen resulting from processing the topological sort is denoted by $\Sigma_{OPT^\star}$. The following proposition explains why the extended algorithm finds best regimen. (Note a ``strong'' sequence of quantifiers---a {\em single} regimen that is optimal for all $X$ of $\a$ when competing with any regimen that start with $X$.)

\begin{proposition}
\label{p.dp}
When the algorithm OPT$^\star$ halts, thus computing a regimen $\Sigma_{OPT^\star}$, for any node $X$ of $\a$, the expected time to completion for the regimen starting with tasks $X$ already executed is equal to the minimum expected time to completion $B_X$ that any regimen can achieve starting with tasks $X$ already executed.
\end{proposition}

\remove{**********************
\begin{proof}[sketch]
We show that as we process the topological sort, we maintain an invariant that for all sets $Z$ subsequent to $Y$ in the sort, regimen $\Sigma_{OPT}$ achieves $B_Z$ when starting with $Z$. When processing a set $Y$, an optimal regimen starting with $Y$ will have, by Theorem~\ref{t.recursive}, expectation expressed as expectations for sets subsequent in the sort. By the invariant they must be at least the corresponding expectations for $\Sigma_{OPT}$. Hence $\Sigma_{OPT}(Y)$ can be selected optimally.
\end{proof}
**********************}

\begin{proof}
Let $Y_1,\ldots,Y_m$ be the topological sort of $\a$ which was used by the algorithm OPT$^\star$. We argue that the following invariant parameterized by $h$ holds during the reverse processing of the sort:
\begin{quote}
For all $i$ such that $h\leq i\leq m$, $S_{Y_i} = B_{Y_i}$ and $S_{Y_i}$ is equal to the expected time to completion of $\Sigma_{OPT^\star}$ that starts with tasks $Y_i$ already executed.
\end{quote}

This invariant is clearly true for $h=m$. Indeed, by Lemma~\ref{l.dag}, $\a$ has just one sink which is equal to $N_{\fg}$, and so $Y_m = N_{\fg}$. If every task is already executed, the minimum expected time to completion is zero. The algorithm, thus, correctly assigns zero to $S_{Y_m}$. The expectation of any regimen that starts with every task already executed is zero, and so the value of $\Sigma_{OPT^\star}(Y_m)$ is satisfactory.

Now pick any $2\leq h\leq m$. We shall see that after OPT$^\star$ has processed $Y_{h-1}$, the invariant is true for $h$ decreased by one. Let $X=Y_{h-1}$. Since $X$ is not a sink of $\a$, it has at least one child. Let $X_1,\ldots,X_k$ be the children.

We first argue that the value assigned to $S_{X}$ is at most the minimum expected time to completion $B_X$.
Pick a regimen $\Sigma$ that achieves the minimum when starting with tasks $X$ already executed. Some workers may get assigned to $\perp$ in $\Sigma(X)$. Lemma~\ref{l.finite} ensures that $B_X$ is finite. We can use Theorem~\ref{t.recursive} to express the expectation of the regimen as a weighted sum 
$$
B_X 
=
 1/(a_1+\ldots+a_k)\cdot(1+\sum_{i=1}^ka_iT_{X_i} )
$$
of expected times $T_{X_i}$ to completion for $\Sigma$ starting with tasks from the sets $X_i$ already executed. By the invariant $S_{X_i}=B_{X_i}$, and so the expectation $T_{X_i}$ is at least $S_{X_i}$. The dynamic program considers every assignment as a candidate for $\Sigma_{OPT^\star}(X)$, and the assignment $\Sigma(X)$ in particular. For this assignment the algorithm calculates the weighted sum $1/(a_1+\ldots+a_k)\cdot(1+\sum_{i=1}^ka_iS_{X_i} )$. Hence the weighted sum is at most $B_X$. The algorithm selects an assignment that minimizes the weighted sum, so the value of $S_X$ is at most $B_X$, as desired.

After values of $S_X$ and $\Sigma_{OPT^\star}(X)$ have been selected, the value of $S_X$ is clearly equal to the expectation of $\Sigma_{OPT^\star}$ starting with $X$ already executed. Indeed, for any assignment considered by the algorithm, the weighted sum is, by Theorem~\ref{t.recursive}, equal to the expected time to completion of $\Sigma_{OPT^\star}$ that starts with tasks $X$ already executed and uses this assignment in place of $\Sigma_{OPT^\star}(X)$.

Since $\Sigma_{OPT^\star}$ has expectation $S_X$ and $S_X$ is at most $B_X$, the expected time to completion for $\Sigma_{OPT^\star}$ starting with $X$ already executed is actually equal to $B_X$. Hence the invariant holds when $h$ gets decreased by one. This completes the proof.
\end{proof}

We will use this result to show that it is not necessary for workers to idle. We begin by observing that, roughly speaking, ``more done, less time remaining to complete'', as detailed in the subsequent lemma.

\begin{lemma}
\label{l.monotone}
For any subsets $X\subseteq Y$ of tasks that satisfy precedence constraints, $B_X \geq B_Y$.
\end{lemma}

\begin{proof}
It is sufficient to prove the lemma with an additional restriction on $Y$. Specifically, the proof considers any $X$ as stated, any $u\in E(X)$, and $Y=X\cup\set{u}$. The proof is by the reverse induction on the cardinality of $X$. The theorem is obvious when $\card{X}=t-1$, because then $B_Y=0$.

For the inductive step, take $X$ of cardinality at most $t-2$. Consider the regimen $\Sigma_{OPT^\star}$. By Proposition~\ref{p.dp} the expected time to completion when $\Sigma_{OPT^\star}$ starts with tasks $Z$ already executed is $B_Z$, for any $Z$ that satisfies precedence constraints. Let $t_1,\ldots,t_k$ be the tasks to which $\Sigma_{OPT^\star}(X)$ assigns workers. Since $B_X$ is finite, $k\geq 1$. Let $p_i$ be the probability that at least one worker assigned to task $t_i$ succeeds. 
Let $p_C=\prod_{i\in C} p_i\prod_{i\in [k]\setminus C} (1-p_i)$ be the probability that exactly the set of tasks indexed by $C$ gets executed in $\Sigma_{OPT^\star}(X)$ among tasks indexed by $[k]$. We can use Theorem~\ref{t.recursive} to express the expected time to completion for $\Sigma_{OPT^\star}$ as
\begin{equation}
\label{eq.B_X}
B_X = \frac{1+\sum_{\emptyset\subset C \subseteq[k]}p_C\cdot B_{X\cup C}}{1-p_{\emptyset}}
\end{equation}
with $p_{\emptyset} < 1$.

We begin with a few special cases that are simpler to analyze and that simplify the follow-up analysis. 
Suppose first that $\Sigma_{OPT^\star}(X)$ does not assign any worker to $u$. Then each assigned task is still eligible when tasks $X\cup\set{u}$ have been executed. Hence we can define a new regimen $\Sigma'$ equal to $\Sigma_{OPT^\star}$ except that $\Sigma'(X\cup\set{u}) = \Sigma_{OPT^\star}(X)$. But then the expected time to completion for $\Sigma'$ starting with tasks $X\cup\set{u}$ already executed is
$$
T'_{X\cup\set{u}} = \frac{1+\sum_{\emptyset\subset C \subseteq[k]}p_C\cdot B_{X\cup C\cup\set{u}}}{1-p_{\emptyset}}
~.
$$
Clearly, by the inductive assumption, $B_{X\cup C} \geq B_{X\cup C \cup\set{u}}$. So $B_X\geq T'_{X\cup\set{u}}$. Since $ T'_{X\cup\set{u}} \geq B_{X\cup\set{u}}$, we have that $B_X \geq B_{X\cup\set{u}}$, as desired. 
Hence in the remainder we can assume that $t_k=u$.
Suppose now that $k=1$. Then $B_X$ is equal to $1/(1-p_{\emptyset})$ plus $B_{X\cup\set{t_k}}$ (because $p_{\emptyset}$ cannot be $1$), and so clearly the latter is less than $B_X$. A similar conclusion follows when $k\geq 2$ but $0=p_1=\ldots=p_{k-1}$. 

After having tackled the simpler special cases, we proceed to the final, most interesting, case when $k\geq 2$, $t_k=u$, and at least one of $p_1,\ldots,p_{k-1}$ is strictly greater than $0$. 
The Equation~(\ref{eq.B_X}) for $B_X$ contains probabilities of events, and it is convenient to partition the events into three groups:
\begin{itemize}
\item[(i)] only task $t_k$ gets executed, which occurs with probability $p_k\cdot (1-p_1)\cdot\ldots\cdot(1-p_{k-1})$,
\item[(ii)] precisely the tasks indexed by a non-empty set $C \subseteq [k-1]$ get executed and also task $t_k$, which occurs with probability $p_k\cdot\prod_{i\in C} p_i\prod_{i\in [k-1]\setminus C} (1-p_i)$,
\item[(iii)] precisely the tasks indexed by a non-empty set $C \subseteq [k-1]$ get executed but not task $t_k$, which occurs with probability $(1-p_k)\cdot\prod_{i\in C} p_i\prod_{i\in [k-1]\setminus C} (1-p_i)$.
\end{itemize}

The Equation~(\ref{eq.B_X}) for $B_X$ has a summation where a probability from the second group is multiplied by $B_{X\cup C\cup\set{t_k}}$, while a probability from the third group by $B_{X\cup C}$. By the inductive assumption, $B_{X\cup C\cup\set{t_k}} \leq B_{X\cup C}$, since $C$ is not empty. Hence the latter can be replaced by the former, and the resulting expression will not increase. But then we can pairwise combine each summand of group (ii) with the corresponding summand of group (iii), and obtain 
$$
B_X \geq \frac{1+p_kq_{\emptyset}\cdot B_{X\cup\set{t_k}} + \sum_{\emptyset\subset C\subseteq [k-1]} q_C\cdot B_{X\cup C\cup\set{t_k}}}{1-(1-p_k)q_{\emptyset}}
~,
$$
where $q_C=\prod_{i\in C} p_i\prod_{i\in [k-1]\setminus C} (1-p_i)$ is the probability that exactly the set of tasks indexed by $C$ gets executed among tasks from $[k-1]$. Let us denote the big summation by $S$. It is now enough to show that 
$$
\frac{1+p_kq_{\emptyset}\cdot B_{X\cup\set{t_k}}
 + S}{1-(1-p_k)q_{\emptyset}}
\geq 
B_{X\cup\set{t_k}}
~.
$$
In order to verify this inequality, we transform it to an equivalent form by multiplying both sides by the denominator and grouping for $B_{X\cup\set{t_k}}$
$$
\frac{1+S}{1-q_{\emptyset}}
\geq 
B_{X\cup\set{t_k}}
~,
$$
because $q_{\emptyset}<1$.
We notice that the left hand side is the expected time to completion of a regimen $\Sigma''$ that starts with tasks $X\cup\set{t_k}$ already executed. This regimen $\Sigma''$ is equal to $\Sigma_{OPT^\star}$ except that when tasks $X\cup\set{t_k}$ need to be executed, the regimen assigns workers to tasks the way $\Sigma_{OPT^\star}(X)$ does, omitting these workers assigned to $t_k$ (such workers do not get assigned then at all). Since $B_{X\cup\set{t_k}}$ is the minimum expectation, it must indeed be smaller or equal to the expectation of that regimen $\Sigma''$.

This completes the inductive step and the proof.
\end{proof}

We use the lemma to argue that when processing $Y_h$, it is sufficient to consider assignments $S$ that map every worker to an eligible task, i.e., we do not need to assign any worker to $\perp$.

\begin{proposition}
For any set $X$ of tasks that satisfies precedence constraints and does not contain every sink of $\g$, there is a regimen $\Sigma$ such that $\Sigma(X)$ assigns every worker to a task of $E(X)$ and the expected time to completion of $\Sigma$ that starts with tasks $X$ already executed is $B_X$ (i.e., minimum).
\end{proposition}

\begin{proof}
We shall construct the regimen $\Sigma$ using the regimen $\Sigma_{OPT^\star}$ of Proposition~\ref{p.dp}. We will show that the expectation of $\Sigma$ is at most that of $\Sigma_{OPT^\star}$, and since the latter is minimum, the expectations will have to be equal. 

Let us first make a few observations about $\Sigma_{OPT^\star}$. Let $t_1,\ldots,t_k$ be the tasks to which $\Sigma_{OPT^\star}(X)$ assigns workers. We know that $k\geq 1$. If every worker has been assigned by $\Sigma_{OPT^\star}(X)$ to a task, then the claim follows. We, therefore, assume that there is at least one worker not assigned by $\Sigma_{OPT^\star}(X)$ to any task. 

We construct a regimen $\Sigma$. Let $U$ be the set of workers not assigned by $\Sigma_{OPT^\star}(X)$ to any task. We define $\Sigma$ to be equal to $\Sigma_{OPT^\star}$ except that the assignment $\Sigma(X)$ also assigns every worker from $U$ to task $t_1$. Let $q$ be the probability that at least one worker from $U$ succeeds when assigned to $t_1$, $0\leq q\leq 1$. We shall argue that the expected time to completion for $\Sigma$ starting with tasks $X$ already executed, $T_X$, is at most $B_X$.

Let us compare equations for $B_X$ and $T_X$. As earlier, let $p_C$ be the probability that exactly the tasks indexed by $C$ get executed. We can express $B_X$ using Theorem~\ref{t.recursive} as
$$
B_X = \frac{1+\sum_{\emptyset \subset C \subseteq [k]}p_C\cdot B_{X\cup C}}{1-p_{\emptyset}}
~,
$$
and we know that $p_{\emptyset}<1$.
Similarly, we can express $T_X$. We notice that for any $C$, $\emptyset \subseteq C \subseteq [k]$, the event that exactly the tasks in $C$ are executed in $\Sigma_{OPT^\star}(X)$ can be decomposed into two mutually exclusive ``subevents'': at least one worker from $U$ succeeds, and every worker from $U$ fails. Hence
$$
T_X = \frac{1+q\cdot p_{\emptyset}\cdot B_{X\cup\set{t_1}}+ \sum_{\emptyset \subset C \subseteq [k]}q\cdot p_C\cdot B_{X\cup C\cup{\set{t_1}}}+\sum_{\emptyset \subset C \subseteq [k]}(1-q)\cdot p_C\cdot B_{X\cup C}}{1-(1-q)\cdot p_{\emptyset}}
~,
$$
even if $t_1 \in C$ for some $C$.

We make two simplifying assumptions. Notice that if $q=0$, then workers of $U$ always fail, and so $T_X$ is then obviously equal to $B_X$. Notice also that if $p_{\emptyset}=0$ then there must be a task that is certainly executed. We can then reassign workers of $U$ to that task instead of to $t_1$, and the resulting $T_X$ will also be equal to $B_X$. Therefore, we can assume without loss of generality that $q>0$ and $p_{\emptyset}>0$.

We will modify the expression for $T_X$ to make it look more like $B_X$. By Lemma~\ref{l.monotone}, we know that $B_{X\cup C\cup\set{t_1}} \leq B_{X\cup C}$. Hence we can replace the former with the latter in the expression for $T_X$, pairwise combine terms, and obtain a bound
\begin{equation}
\label{eq.T_X}
T_X \leq \frac{1+q\cdot p_{\emptyset}\cdot B_{X\cup\set{t_1}}+ \sum_{\emptyset \subset C \subseteq [k]} p_C\cdot B_{X\cup C}}{1-(1-q)\cdot p_{\emptyset}}
~,
\end{equation}
that begins to resemble the expression for $B_X$. 

We shall conclude that the bound in Equation~(\ref{eq.T_X}) is at most $B_X$, which will mean that $T_X\leq B_X$, as desired, because then of course $T_X=B_X$. Let $S$ denote the summation in the enumerator. We restate our goal as that of showing that
$$
\frac{1+q\cdot p_{\emptyset}\cdot B_{X\cup\set{t_1}}+ S}{1-(1-q)\cdot p_{\emptyset}}
\leq
B_X
=
\frac{1+S}{1-p_{\emptyset}}
~.
$$
But this inequality can be equivalently expressed as
$$
B_{X\cup\set{t_1}} \leq \frac{B_X\bra{1-p_{\emptyset}+q\cdot p_{\emptyset}}-1-S}{q\cdot p_{\emptyset}}
= \frac{1+S}{1-p_{\emptyset}}\bra{\frac{1-p_{\emptyset}}{q\cdot p_{\emptyset}}+1}- \frac{1+S}{q\cdot p_{\emptyset}}= \frac{1+S}{1-p_{\emptyset}} = B_X
~,
$$
because $q,p_{\emptyset}>0$. However, by Lemma~\ref{l.monotone}, $B_{X\cup\set{t_1}}$ is indeed at most $B_X$, and thus the proof is completed.
\end{proof}

\begin{corollary}
\label{dp.correctness}
The algorithm $OPT$ finds a regimen $\Sigma_{OPT}$ that minimizes the expected completion time.
\end{corollary}

\begin{figure*}[ht!]
\centerline{\epsfig{file=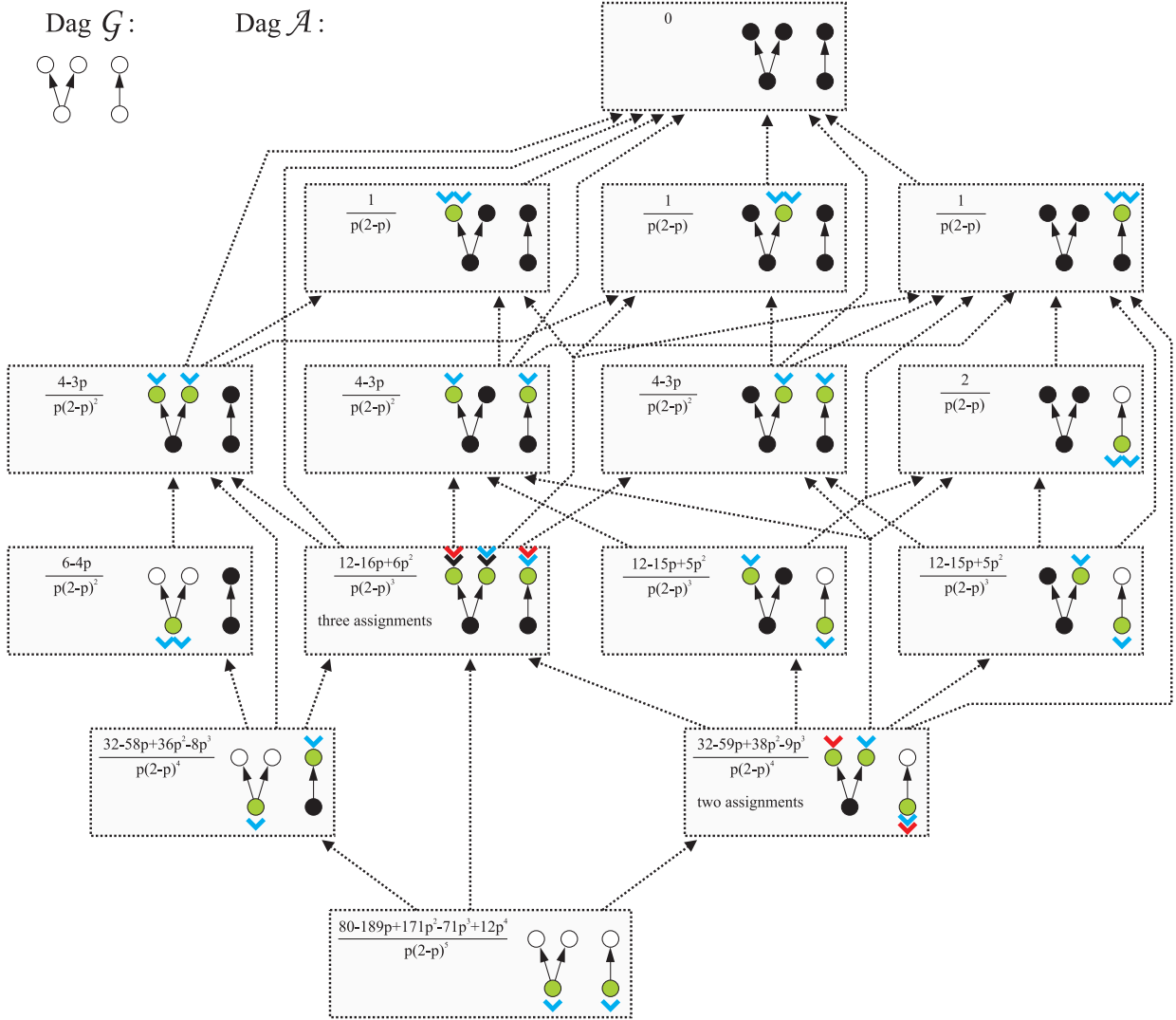,width=6.45in}}
\caption{Depicted is the dag $\a$, called admissible evolution of execution, constructed for the given dag $\g$. The dashed boxes and dashed arrows represent nodes and arcs of $\a$. Depicted are also optimal regimens for the case of two workers, where each has the same probability $0<p<1$ of executing successfully any given task. Each box contains: a subset of executed tasks of $\g$ (black), and tasks that are eligible then (green); the minimum expected time to completion (upper left hand corner); and how the two workers should be assigned to eligible tasks so as to yield the expectation (same-color ``$\vee$'' signs)---every assignment that yields this expectation is shown, except for symmetric assignments where workers could be swapped. Note that in the first round the two workers must be assigned to distinct tasks, if not then the expected completion time is not minimized.
\label{f.aec}}
\end{figure*}

We now tackle the second subgoal of showing that the algorithm runs in polynomial time for dags of at most a constant width and when the number of workers is at most a constant.
Let $w$ be the width of $\g$. By Dilworth's Theorem, the dag $\g$ can be covered by $w$ chains. If any task of the chain is executed, so must be every preceding task in the chain. Each chain can have length at most $t$. Hence, there are at most $(t+1)^{w}$ distinct subsets of tasks that satisfy precedence constraints. 
When the algorithm processes a node of $\a$, it considers assignments of workers to eligible tasks. Any worker can, therefore, be assigned in at most $t$ ways, and so there are at most $t^n$ assignments. We notice that for a given assignment, we actually only need to evaluate probabilities $a_i$ for at most $2^n$ of the sets $D_i$, because the workers can get assigned to at most $n$ different tasks, and so probabilities associated with sets that are not subsets of these tasks must be zero, and so can be omitted from the weighted sum. This yields the main result of this section.

\begin{theorem}
\label{dp.final}
The algorithm OPT solves the ROPAS Problem in polynomial time when dag width is bounded by a constant and the number of workers is also bounded by a constant.
\end{theorem}

An example of an application of the algorithm to a dag is given in Figure~\ref{f.aec}. 

\section{Complexity of scheduling}
\label{s.complexity}

This section shows fundamental reasons why the algorithm presented in the previous section is so restricted. Recall that the algorithm restricts both the number of workers {\em and} dag width to at most a constant. It turns out that ROPAS is NP-hard if the number of workers can grow while dag width is at most a constant, {\em or} if dag width can grow while the number of workers is at most a constant. Therefore, a polynomial time algorithm for the problem must restrict both dag width {\em and} the number of workers, unless $P\neq NP$. We also show that if we allow the number of workers to grow and at the same time dag width to grow, then the problem is inapproximable within factor less than $5/4$, unless P=NP.

Our reductions demonstrate that the complexity of the ROPAS Problem is quite steep. 
First, the problem is trivial when the dag has just one task. Simply then an optimal regimen assigns every worker to that task. However, the problem becomes NP-hard when dag has just two independent tasks! (Arguably this is the ``second simplest'' dag.) Intuitively, the hardness results from the difficulty to perfectly load balance unreliable workers across eligible tasks.
Second, the problem is easy when there is just one worker. Simply then an optimal regimen follows any topological sort and assigns the (only) worker to the task at hand. However, the problem becomes NP-hard when there are just two workers! Intuitively, here hardness is a consequence of possible complex task dependencies which make it difficult to determine how to pick tasks for execution so that sufficiently many tasks become eligible thus preventing workers from compulsory idling.

\subsection{Scheduling narrow dags with many workers is NP-hard}
\label{s.narrow_many}

In this section we show that ROPAS restricted to the dag with two independent tasks, but where the number of workers can grow, is NP-hard. For this purpose we introduce two auxiliary problems and determine their complexities. The first of them is the Multiplicative Partition Problem.

\FF
\noindent
{\bf Multiplicative Partition (MP)}\\
{\it Instance:} Number $n$, integer sizes $s_1,\ldots,s_n\geq 2$.\\
{\it Question:} Is there a set $P$ such that 
$
\prod_{i\in P} s_i = \prod_{i\notin P} s_i
$~?
\FF

It is tempting to try to show NP-completeness of MP by a reduction from the well-known Partition Problem. Here we could hope to use the fact that 
$\sum_{i\in P} r_i = \sum_{i\notin P} r_i$
if, and only if, 
$\prod_{i\in P} 2^{r_i} = \prod_{i\notin P} 2^{r_i}$. Unfortunately, an $s_i=2^{r_i}$ takes space exponential with respect to the space taken by $r_i$. The issue of space explosion can be avoided by carrying out a more subtle reduction from the Exact Cover by 3-Sets (X3C) Problem that uses certain 
fundamental
facts from Number Theory.

\begin{lemma}
The MP Problem is NP-complete.
\end{lemma}

\begin{proof}
We will show a reduction from the Exact Cover by 3-Sets Problem to the Multiplicative Partition Problem.

Pick any instance of X3C---a collection $C=\set{S_1,\ldots,S_n}$ of $3$-element subsets of the set $U=\set{1,\ldots,3m}$, $S_i=\set{a_i,b_i,c_i}$. Recall that the instance is positive if, and only if, there is a subcollection of $C$ such that every element of $U$ occurs in exactly one subset of the subcollection. We can assume that $m\leq n$ and that the union of every subset of $C$ is $U$, as otherwise the instance of X3C is definitely negative.
We build an instance of MP as follows. We begin by enumerating the first consecutive primes $p_1,\ldots,p_{3m+1}$. By the Prime Number Theorem (see e.g., \cite{IR90}) the prime $p_{3m+1}$ has value $O(m \ln m)$. Therefore, we can enumerate the primes in time $O(m^2\log^2 m)$ using the Eratosthenes sieve; this time is polynomial with respect to the size of the instance of X3C. 
The instance of the MP Problem will have $n+2$ sizes. The first $n$ are defined as products of primes indexed by the elements of subsets i.e., $s_i=p_{a_i}\cdot p_{b_i}\cdot p_{c_i} \geq 2$. Note that $s_i$'s are fairly small, so that each $s_i$ is $O(n^6)$. 
We add two ``dummy'' sizes. Let $p$ be the product of sizes $s_1$ through $s_n$, and $q$ be the product of the first $3m$ primes; each is $O(n^{6n})$, so sufficiently small. Note that $p/q$ is an integer, because $\bigcup_i S_i=U$. We define $s_{n+1} = p_{3m+1}\cdot p/q$ and $s_{n+2} = p_{3m+1}\cdot q$. Therefore, the instance of the MP Problem can be represented in time and space polynomial with respect to $n$, as desired.

Suppose that the instance of X3C is positive, and let $Q$ be the indices that define the subcollection. But then $\prod_{i\in Q} s_i$ is the product of all primes from $p_1$ to $p_{3m}$. So then $\prod_{i\notin Q} s_i$ is equal to the product of the $s_i$'s divided by the product of these primes. Therefore,
$$
s_{n+1}\prod_{i\in Q} s_i
=
s_{n+2}\prod_{i\notin Q} s_i 
~,
$$
and so the instance of MP is positive.

Suppose that the instance of MP is positive, and let $P$ be a subset such that $\prod_{i\in P} s_i = \prod_{i\notin P} s_i$. Note that there are only two sizes, $s_{n+1}$ and $s_{n+2}$, that have the prime $p_{3m+1}$ as a factor. But two natural numbers are equal if, and only if, they have the same factorization. Thus in order for the products to be equal, either $n+1$ or $n+2$ is in $P$, but not both. We can assume, without loss of generality, that $n+1$ is in $P$.
Let $L=P\setminus\set{n+1}$, and $R$ be the elements not in $P$ and other than $n+2$. But then 
$$
s_{n+1}\cdot\prod_{i\in L} s_i = s_{n+2}\cdot \prod_{i\in R} s_i
~,$$
which can be equivalently expressed as
$$
p/q\cdot\prod_{i\in L} s_i = 
q\cdot \prod_{i\in R} s_i
~.
$$
Of course, $L$ and $R$ partition $\set{1,\ldots,n}$, so $p=\prod_{i\in L} s_i
\cdot 
\prod_{i\in R} s_i
$. Thus
$
q^2 = 
\bra{\prod_{i\in L} s_i}^2
$, and since the numbers are positive, $q = 
\prod_{i\in L} s_i$. Therefore, there is a selection of the $s_i$'s whose product contains each prime from $p_1$ to $p_{3m}$ exactly once. This set $L$ defines a subcollection of $C$ such that every element of $U$ is contained in exactly one subset of the subcollection. So the instance of X3C is positive.
\end{proof}

The MP Problem is used to demonstrate NP-hardness of the second auxiliary problem called Partition with Minimum Sum of Products.

\FF
\noindent
{\bf Partition with Minimum Sum of Products (PMSP)}\\
{\it Instance:} Number $n$, rational numbers  $0<r_1,\ldots,r_n<1$.\\
{\it Objective:} Find a set $P$ that minimizes
$
\prod_{i\in P} r_i + \prod_{i\notin P} r_i
~.
$ 
\FF

\begin{lemma}
The PMSP Problem is NP-hard.
\end{lemma}

\remove{**********************
\begin{proof}[Proof sketch.]
The proof is based on the fact that for any $d>0$, the function $f(x) = e^x+e^{d-x}$ defined for $0\leq x\leq d$ is minimized exactly when $x=d/2$. We avoid using irrational numbers by selecting the $r_i$ appropriately.
\end{proof}
**********************}

\begin{proof}
We give a polynomial time Turing reduction of MP to PMSP. Specifically, we show how to find an answer to a given instance of MP by inspecting the minimum value for an appropriately constructed instance of PMSP. The proof is based on the fact that for any $d>0$, the function $f(x) = e^x+e^{d-x}$ defined for $0\leq x\leq d$ is minimized exactly when $x=d/2$. The difficulty is to ensure that we compute using rational numbers only.

Pick any instance of the MP Problem and let $s_1,\ldots,s_n$ be the positive integer sizes, each at least $2$. We define an instance of the PMSP Problem by taking $r_i=h_i^2$, where $h_i=1/s_i$. Then $0<r_i<1$, as desired. 
We shall see that we can decide exactly when the instance of MP is positive, just by checking if the minimum for the instance of PMSP is equal to the value $2\prod_{i=1}^n h_i$.

Let us see when the product $\prod_{i\in P} r_i + \prod_{i\notin P} r_i$ achieves that value. Let $d=\sum_{i=1}^n\ln r_i$. Then
$$
\prod_{i\in P} r_i + \prod_{i\notin P} r_i
=
e^{\sum_{i\in P} \ln r_i} + e^{d-\sum_{i\in P} \ln r_i }
\geq 2e^{d/2}
$$
with equality only when $d/2 = \sum_{i\in P} \ln r_i$. The crucial observation is that $2e^{d/2}$ is just $2\prod_{i=1}^n h_i$, which is a rational number. Thus we can invoke an oracle that solves PMSP on the instance that we have constructed, then inspect the minimum returned by the oracle, and determine whether or not there exists $P$ such that $d/2 = \sum_{i\in P} 2\ln h_i$. But the equation can be equivalently written as
$1/2 \sum_{i=1}^n \ln h_i = \sum_{i\in P} \ln h_i$, which is equivalent to $\sum_{i\in P} \ln h_i = \sum_{i\notin P} \ln h_i$. This in turn is equivalent to $\prod_{i\in P} h_i = \prod_{i\notin P} h_i$. But this is equivalent to 
$\prod_{i\in P} s_i = \prod_{i\notin P} s_i$. So 
the instance of MP is positive if, and only if, the minimum for the instance of PMSP is equal to 
$2\prod_{i=1}^n h_i$.
\end{proof}

We have prepared a toolset, and now we are ready to show the main hardness result of this subsection.

\begin{theorem}
The ROPAS Problem restricted to the dag with two independent tasks (where the number of workers may grow) is NP-hard.
\end{theorem}

\begin{proof}
We give a polynomial time Turing reduction of PMSP to ROPAS. Pick any $0<r_1,\ldots,r_n<1$, and let us consider a dag with just two isolated nodes, call them left task and right task. There are $n$ workers. We define the probability of success of worker $i$, $p_{i,left}$ and $p_{i,right}$, to be $1-r_i$, and denote it simply by $p_i$. Pick a regimen $\Sigma$ that minimizes the expected completion time. Let $L$ be the set of workers assigned by $\Sigma(\emptyset)$ to the left task i.e., assigned when both tasks are eligible. We shall argue that the expected completion time is minimized if, and only if, $\prod_{i\in L} r_i +\prod_{i\notin L} r_i$ is minimized. Therefore, a best regimen can be transformed to a solution to the instance of PMSP.

In order to show the equivalence, we make a sequence of observations about the assignments that the regimen must make. Let $R$ be the workers assigned by $\Sigma(\emptyset)$ to the right task. Hence $L$ and $R$ are disjoint and included in $\set{1,\ldots,n}$. Let $\ell=1-\prod_{i\in L}(1-p_i)$ be the probability that at least one left worker succeeds, and $r=1-\prod_{i\in R}(1-p_i)$ that at least one right succeeds. We can use Theorem~\ref{t.recursive} to calculate the expected completion time $T_{\emptyset}$ 
\begin{align*}
T_{\emptyset}  = &
\frac{1+
\ell(1-r) \cdot T_{\set{left}}
+
(1-\ell)r \cdot T_{\set{right}}}{1-(1-\ell)(1-r)}~,
\end{align*}
where $T_{\set{left}}$ is the expected time to completion by the regimen starting with the left task already executed, and $T_{\set{right}}$ starting with the right task already executed. 

We can simplify this expression by observing that $T_{\set{left}}$ and $T_{\set{right}}$ are equal. Indeed, treating sets $L$ and $R$ as constants, we notice that $T_{\emptyset}$ is minimized when $T_{\set{left}}$ and $T_{\set{right}}$ are minimized. Let $K\subseteq \set{1,\ldots,n}$ be the set of workers assigned to the left task in $\Sigma(\set{right})$ i.e., when the right task has been executed and the left one is the only eligible task. Then the expected time to completion  $T_{\set{right}}$ is $1/\bra{1-\prod_{i\in K }\bra{1-p_i}}$. Since $0<p_i<1$, this expectation is clearly minimized if, and only if, $K=\set{1,\ldots,n}$. We can apply the same argument to conclude that the regimen assigns every worker to the right task when the task is the only one remaining to be executed. Therefore, $T_{\set{left}}=T_{\set{right}}$. Let us denote this expectation by $T$. Observe that $T>0$. Hence 
$$
T_{\emptyset} = 
\frac{1+T(\ell(1-r)+(1-\ell)r)}{1-(1-\ell)(1-r)}
=
T+\frac{1-T\ell r}{\ell+r-\ell r}
~.
$$

We notice that $\Sigma(\emptyset)$ assigns every worker to a task.
 In order to show this, we demonstrate that the expectation decreases as $\ell$ or $r$ increase. We notice that the formula is symmetric with respect to $\ell$ and $r$, so we only calculate its partial derivative with respect to $r$
$$
\diff{}{r}
=
\frac{-1+\ell - T\ell^2}{(\ell+r-\ell r)^2} 
~.
$$
This derivative is strictly negative for any $0<\ell<1$. So the expectation is minimized when $r$ and $\ell$ are maximized. This is equivalent to the condition that every worker must be assigned; in other words, the sets $L$ and $R$ partition the set $\set{1,\ldots, n}$. 

Let us summarize our observations so far: the regimen initially assigns every worker to a task, and when one task remains to be executed, the regimen assigns every worker to the task.

With these observations, we turn back to determining under what circumstances $T_{\emptyset}$ is minimized. Since $L$ and $R$ partition $\set{1,\ldots,n}$, then $\ell+r-\ell r$ is just the probability that at least one of the $n$ workers succeeds, the value of which is independent of how workers are assigned to left and right task. Thus the expectation is minimized if, and only if, $\ell r$ is maximized across partitions $L$ and $R$ of $\set{1,\ldots,n}$, since $T$ is now a constant. But
\begin{align*}
\ell r =&
1-\bra{\prod_{i\in L}(1-p_i)}
-\bra{\prod_{i\in R}(1-p_i)} +\\
&\bra{\prod_{i\in L}(1-p_i) \prod_{i\in R}(1-p_i)}
~,
\end{align*}
and since $L$ and $R$ is a partition, the expression in the third big parenthesis does not depend on the choice of $L$ and $R$. So the expression is maximized, for a partition $L$ and $R$ that minimizes
$
\prod_{i\in L}r_i
+\prod_{i\in R}r_i
~.$

We summarize the observations that we have made so far. A regimen minimizes expectation if, and only if, its set $L$ minimizes $\prod_{i\in L}r_i
+\prod_{i\notin L}r_i$ and the set $K=\set{1,\ldots,n}$. Thus the theorem is proven.
\end{proof}

\subsection{Scheduling wide dags with few workers is NP-hard}
\label{s.wide_few}

We show that the ROPAS Problem restricted to two workers, but where dag width can grow, is NP-hard. Toward this end, we first establish NP-completeness of an auxiliary problem. Using this problem, we establish NP-completeness of other problem. Finally, that other problem is used in a polynomial time Turing reduction that yields NP-hardness of the restricted ROPAS Problem.

We first demonstrate how to modify a reduction of Johnson~\cite{Joh87} to fit our purpose. The reduction is from a variant of Clique to the Balanced Complete Bipartite Subgraph Problem (cf.~\cite{Pee03}). The reduction takes a graph on $2g$ nodes and $e$ edges, and produces a bipartite graph with $\binom{g}{2}+g$ nodes on the left and $e\leq \binom{2g}{2}$ nodes on the right. The question of finding a $g$-clique in the original graph is show to be equivalent to that of finding a balanced complete bipartite subgraph with $\binom{g}{2}$ nodes on the left and the same number of nodes on the right. Since $\binom{g}{2}+g \leq \binom{2g}{2}$, we can pad the bipartite graph with isolated nodes for uniformity, and then the resulting problem is NP-complete.

\FF
\noindent
{\bf Fixed Ratio Balanced Complete Bipartite Subgraph (FIRBCBS)}\\
{\it Instance:} Number $g$, bipartite graph $\g$ with $\binom{2g}{2}$ nodes on the left and $\binom{2g}{2}$ on the right.\\
{\it Question:} Does $\g$ contain a balanced complete bipartite subgraph with $\binom{g}{2}$ nodes on the left and $\binom{g}{2}$ on the right?
\FF

\begin{lemma}[\cite{Joh87}]
The FIRBCBS Problem is NP-complete.
\end{lemma}

The FIRBCBS Problem can be reduced to an auxiliary problem of finding many subsets whose union is small. 

\FF
\noindent
{\bf Fixed Ratio Many Subsets with Small Union (FIRMSSU)}\\
{\it Instance:} Number $k$, nonempty subsets $S_1,\ldots,S_{3k}$ of the set $\set{1,\ldots,3k}$ whose union is $\set{1,\ldots,3k}$.\\
{\it Question:} Are there $2k$ of these subsets whose union has cardinality at most $2k$?
\FF

The following lemma generalizes an earlier result of~\cite{GaoM04}.
We can prove the lemma using FIRBCBS and padding.

\begin{lemma}
The FIRMSSU Problem is NP-complete.
\end{lemma}

\begin{proof}
Take any instance of FIRBCBS. We consider the $\binom{2g}{2}$ by $\binom{2g}{2}$ matrix $A=(a_{i,j})$ that is the ``complement'' of the adjacency-matrix of $\g$ i.e., $a_{i,j}=0$ if left node $i$ is linked to right node $j$, and $1$ if they are not linked. Its rows and columns can be (independently) rearranged so that the $\binom{g}{2}$ by $\binom{g}{2}$ submatrix in the upper left hand corner contains only zeros if, and only if, the instance is positive. 
Now we will pad $A$ by adding ``strips'' of certain widths (a sketch is in Figure~\ref{f.padding}). Pick the smallest $k$ such that $\binom{2g}{2} < k$. Notice that for such $k$, $\binom{g}{2}<k$. Let $a=k-\binom{g}{2}$ and $b=3k-\binom{2g}{2}-a-k$. This implies that $b>\binom{g}{2}\geq 1$. We then add $a$ rows and $a+k$ columns, every filled with zeros, then $k+b$ rows and $b$ columns, every filled with ones. 
The resulting matrix $A'$ has the property that its rows and columns can be (independently) rearranged so that the $k$ by $2k$ rectangle in the upper left hand corner contains only zeros if, and only if, the instance is positive. The matrix also has at least one row and one column containing only ones, because $b\geq 1$.

We define an instance of FIRMSSU such that each column of $A'$ is a characteristic vector of a subset. This clearly yields a desired instance and completes the reduction.
\end{proof}

\begin{figure}[htb]
\centerline{\epsfig{file=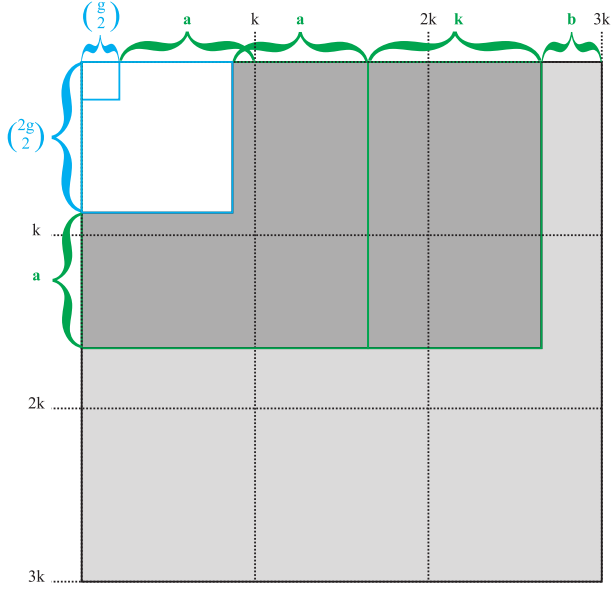,width=3.2in}}
\caption{A padding pattern in the reduction from FIRBCBS to FIRMSSU; light grey areas correspond to ones, dark grey to zeros.
\label{f.padding}}
\end{figure}

We are now ready to prove the main hardness result of the subsection.

\begin{theorem}
\label{t.two}
The ROPAS Problem restricted to two workers (where the dag width may grow) is NP-hard.
\end{theorem}

\begin{proof}
We give a polynomial time Turing reduction from FIRMSSU to ROPAS---by checking if the minimum expectation is small enough we can determine exactly when a given instance of FIRMSSU is positive. 

Take any instance of FIRMSSU with $n=3k$ nonempty subsets $S_1,\ldots,S_{3k}$ on $\set{1,\ldots,3k}$.
We construct an instance of ROPAS with two workers. The dag will have three levels (an example is in Figure~\ref{f.reduction_wide}).
We start with a bipartite dag with $n^2$ sources $A_{1}$ and $n$ sinks $B_{2}$. Sources are partitioned into $n$ groups of cardinality $n$ each. There is an arc leading from a source of group $i$ to a sink $j$ exactly when $i$ is in the set $S_j$. 
There are $2/3n^2$ extra sources $A_2$. Arcs lead from every extra source to every sink $B_{2}$. 
There are $2/3n$ internal tasks $B_1$. Arcs lead from every source of $A_2$ to every internal task of $B_1$. 
There are $1/3n^2$ extra sinks $C_2$. Arcs lead from every internal task of $B_1$ to every extra sink of $C_2$. 
The first worker is certain to execute every task from $A_1$ and $B_1$, while the second worker fails on any. The second worker is certain to execute every task from $A_2$, $B_2$ and $C_2$, while the first worker fails on any.

\begin{figure}[htb]
\centerline{\epsfig{file=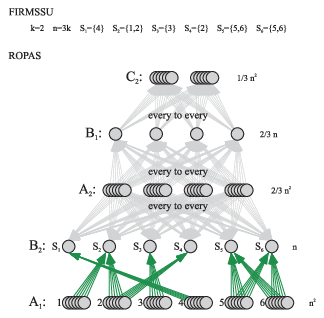,width=3.3in}}
\caption{An example of a reduction from FIRMSSU to ROPAS with two workers and unconstrained dag width. The dag has three levels denoted by $A$, $B$ and $C$. Tasks in sets $A_1$ and $B_1$ are executed only by the first worker; tasks in sets $A_2$, $B_2$ and $C_2$ only by the second worker.
\label{f.reduction_wide}}
\end{figure}

Suppose that the instance of FIRMSSU is positive i.e., there are $2k$ subsets whose union is at most $2k$. Let $U$ be the union of the $2k$ subsets---if $U$ has cardinality less than $2k$, add to $U$ arbitrary elements of $\set{1,\ldots,3k}$ so as to increase the cardinality of $U$ to $2k$. We construct a regimen. The second worker executes sources $A_2$ during the first $2/3n^2$ rounds. Hence the $2/3n$ internal tasks of $B_1$ become eligible. Meanwhile, the first worker executes the $2/3n^2$ sources of $A_1$ from the groups corresponding to $U$. As a result, at least $2/3 n$ sinks of $B_{2}$ become eligible. Following round $2/3n^2$, the second worker begins executing the eligible sinks of $B_{2}$, and computes $2/3n$ of them by the end of round $2/3n^2 + 2/3n$. Meanwhile, the first worker executes the $2/3n$ internal tasks $B_1$, which makes every of the $1/3n^2$ sinks of $C_2$ eligible. Following round $2/3n^2 + 2/3n$, the second worker executes $1/3n^2$ sinks of $C_2$, while the first worker resumes executing the remaining $1/3 n^2$ sources of $A_1$. Thus by the end of round $n^2 + 2/3n$, the remaining $1/3n$ unexecuted sinks of $B_{2}$ are eligible. Finally, the second worker executes these sinks. Hence the computation is completed at the end of round $n^2 +n$.

Suppose now that the instance of FIRMSSU is negative. Pick any regimen. For the reason that the instance is negative, the union of any $2k$ subsets has cardinality at least $2k+1$. But the cardinality of any group is $n$, so at least $2/3n^2+n$ sources of $A_{1}$ must be executed in order to have $2/3n$ sinks of $B_{2}$ eligible. Therefore, at most $2/3n-1$ sinks of $B_{2}$ can get executed by the end of round $2/3n^2+2/3n$. By the end of this round, no sinks of $C_2$ can be executed, because each depends on the execution of $2/3n^2+2/3n$ other tasks. Hence, by the end of that round, there are at least $1/3n+1 + 1/3n^2$ unexecuted tasks that can only be executed by the second worker. As a result the regimen takes at least $n^2+n+1$ rounds to complete computation.
\end{proof}

\subsection{Scheduling wide dags with many workers is inapproximable}

We have seen that ROPAS is NP-hard either when dag can have large width or when the number of workers can be large. It is, therefore, interesting to see how hard the problem is when not only dag width but also the number of workers can be large. Here we show that then ROPAS is inaproximable with a factor less than $5/4$, unless P=NP.

\begin{theorem}
The ROPAS Problem cannot be approximated with factor less than $5/4$, unless P=NP.
\end{theorem}

\begin{proof}
The proof constructs a dag similar to that in the proof of Theorem~\ref{t.two}. Since we can afford many workers, the execution of the dag can be completed in as few as four rounds under favorable circumstances. When circumstances are not favorable, any regimen needs five rounds or more because some workers will necessarily idle. This yields a desired gap.

Specifically, we show a gap creating reduction from FIRMSSU to ROPAS.
Take any instance of FIRMSSU with $n=3k$ nonempty subsets $S_1,\ldots,S_{3k}$ on $\set{1,\ldots,3k}$. 
We construct an instance of the ROPAS Problem as follows. 
There are $4k$ workers. For any task and any worker, the probability that the worker succeeds in executing the task is one. The dag has three levels (an example is in Figure~\ref{f.reduction}). There are $6k$ sources, denoted by $A$, partitioned into $3k$ groups, associated with numbers $1$ through $3k$, of two tasks each. There are $3k$ internal tasks, denoted by $B_a$, associated with numbers $1$ through $3k$. Any source from $A$ from group $i$ is linked to the internal task $i$ from $B_a$. In addition, there are $k$ extra internal tasks, denoted by $B_b$. Each source from $A$ is linked to every internal task from $B_b$. There are $3k$ sinks, denoted by $C_a$, associated with subsets $S_1$ through $S_{3k}$. Any internal task $i$ from $B_a$ is linked to sink $S_j$ if, and only if, $i$ is in the subset $S_j$. In addition, there are $3k$ extra sinks, denoted by $C_b$. Each internal task from $B_b$ is linked to every sink from $C_b$.
We shall see that there is a gap in the minimum expected time to completion for the instance of ROPAS, that differentiates the case when the instance of FIRMSSU is positive from the case when it is negative.

Suppose that the instance of FIRMSSU is positive i.e., there are $2k$ of the subsets whose union has cardinality at most $2k$. We will argue that the minimum expected time to completion is at most $4$. Let $U$ be the union of the $2k$ subsets---if $U$ has cardinality less than $2k$, add to $U$ arbitrary elements of $\set{1,\ldots,3k}$ so as to increase the cardinality of $U$ to $2k$. The set $U$ represents $2k$ internal tasks from $B_{a}$. 
We describe a regimen. In the first round, we execute the $4k$ sources from $A$ the children of which are the $2k$ internal tasks from $B_a$ associated with tasks of $U$. This makes the tasks associated with $U$ eligible. In the second round, we execute the $2k$ eligible tasks from $B_a$. Because the instance is positive, $2k$ sinks from $C_a$ become eligible. In addition, we execute the remaining $2k$ sources from $A$. Because now every source is executed, the remaining $2k$ internal tasks are eligible at the end of the second round. In the third round, we execute the $2k$ eligible sinks from $C_a$ and the $2k$ eligible internal tasks from $B_a$ and $B_b$. Because now every internal task is executed, any non-executed sink is eligible. In the fourth round, we executed the remaining $4k$ sinks.

Suppose that the instance of FIRMSSU is negative i.e., the union of any $2k$ of the subsets has cardinality $2k+1$ or more. Take an arbitrary regimen. We will argue that any execution must take at least five rounds. 
At the end of the first round, at most $2k$ internal tasks can be eligible, and these must be tasks from $B_a$. Indeed, each task from $B_b$ is connected to each of the $6k$ sources from $A$, while at most $4k$ of them have been executed; also any internal task from $B_a$ is connected to two distinct sources.
Therefore, in round two, the only internal tasks that can get executed are at most $2k$ tasks from $B_a$. Because the instance is negative, any $2k$ subsets have cardinality $2k+1$ or more, and so one would have to execute at least $2k+1$ tasks from $B_a$ to make $2k$ sinks from $C_a$ eligible. Thus at most $2k-1$ sinks from $C_a$ can become eligible at the end of the second round. But there are $6k$ sinks in total, so at least $4k+1$ of them are not eligible then.
In round three, some nodes get executed, but none of the $4k+1$ sinks. Even if every of these sinks is eligible at the end of round three, then in round four at most $4k$ of them can get executed, leaving one unexecuted at the end of round four. 
So the regimen requires at least five rounds to complete computation.
\end{proof}

\begin{figure}[htb]
\centerline{\epsfig{file=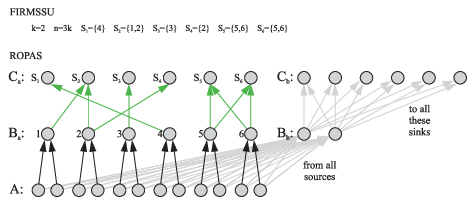,width=6in}}
\caption{Example of reduction from FIRMSSU to ROPAS with unconstrained number of workers and dag width. The dag has three levels denoted by $A$, $B$ and $C$. Every worker can execute any task.
\label{f.reduction}}
\end{figure}

\FF
{\bf Acknowledgements.}
The author thanks Bruce R. Barkstrom of NASA Langley Research Center and Jennifer M. Schopf of Argonne National Laboratory for motivating this study, and Arnold L.~Rosenberg of University of Massachusetts Amherst for discussions on dag scheduling. One missing part of ``the square'' was proven by inspiration of Xiuning Le.
\FF

\end{document}